\documentclass[twocolumn]{article}
\usepackage{amsmath, amssymb, amsthm}
\usepackage{braket}
\usepackage{graphicx}
\usepackage{hyperref}
\usepackage{authblk}
\usepackage{multicol}
\usepackage{cuted} 
\usepackage{mathrsfs}
\usepackage{fancyhdr}
\usepackage{booktabs}
\usepackage{multirow}
\usepackage{enumitem}
\usepackage[margin=2.5cm]{geometry}
\usepackage[numbers,square]{natbib} 
\newtheorem{theorem}{Theorem}

\newtheorem{definition}{Definition}

\newcommand{\mywide}[2][0.9]{\begin{strip}
\noindent\parbox{\textwidth}{#2}
\end{strip}}

\title{Entanglement as a Witness of Quantum Coherence: \\ 
A Bipartite Monty-Hall Protocol}

\author{Jorge Meza-Dom{\'\i}nguez%
\thanks{E-mail: \href{mailto:jorge.meza@cinvestav.mx}{jorge.meza@cinvestav.mx}}}
\affil{Departamento de F{\'\i}sica, Centro de Investigaci{\'o}n y de 
Estudios Avanzados del Instituto Polit{\'e}cnico Nacional, 
Av. Instituto Polit{\'e}cnico Nacional 2508, San Pedro Zacatenco, 
M{\'e}xico 07360, CDMX.}
\date{\today}

\begin{document}

\maketitle

\mywide{
\begin{abstract}
We present a bipartite protocol inspired by the Monty Hall puzzle
that operationally distinguishes quantum coherence from classical
ignorance. A principal qutrit is entangled with an ancillary qutrit
via a controlled unitary, preparing $|\Psi\rangle = 
\frac{1}{\sqrt{3}}(|A,0\rangle + |B,1\rangle + |C,2\rangle)$. A
rank-1 projective discard then eliminates one basis state, leaving
a coherent superposition of the two remaining states. Finally, the
ancilla and qutrit are measured, yielding joint probabilities that
encode the interplay between superposition and measurement
back-action.

We show that the conditional probability $P(B|\text{anc}=0)$ takes
the value $1/4$ in both quantum mechanics and the classical
ignorant-host model, making it unsuitable as a witness. The true
quantum-classical separation emerges in conditional joint
probabilities that correlate ancilla outcomes with specific discard
operations. We define witnesses $\mathcal{W}_{i,j} = P(anc=i,
qutrit=j \mid \text{discard } k)$ where $j$ differs from the
ancilla-implied state. Quantum mechanics predicts $\mathcal{W} =
1/4$, while any classical epistemic model with perfect initial
correlations yields $\mathcal{W} = 0$.

We provide the explicit $9 \times 9$ unitary matrix, a complete
analysis of all measurement outcomes, and a detailed proof of the
violation. The witness is fully immune to white noise and robust
against moderate dephasing. The protocol requires only a single
pair of entangled qutrits and sequential measurements---no spatial
separation, no multiple copies, and no complex sets of
incompatible observables. This makes it suitable for advanced
undergraduate laboratories and provides a pedagogically accessible
test of the ontic-epistemic distinction in quantum foundations.
\end{abstract}
}

\section{Introduction}

\subsection{The ontological status of the quantum state}

A central question in the foundations of quantum mechanics concerns 
the ontological status of the quantum state vector $|\psi\rangle$. 
Does it represent a physical reality existing independently of 
observation (the \emph{ontic} interpretation), or does it merely 
encode an observer's incomplete knowledge about an underlying 
deterministic reality (the \emph{epistemic} interpretation)? 
\cite{Harrigan2010, Spekkens2005}

The epistemic view, tracing its lineage to Einstein's statistical 
interpretation, holds that quantum probabilities arise from ignorance, 
much like classical probabilities in statistical mechanics. A qubit 
in the state $\frac{1}{\sqrt{2}}(|0\rangle + |1\rangle)$ would, in 
this view, be analogous to a coin that is definitely either heads or 
tails but whose identity we do not know. The superposition reflects 
our knowledge, not the physical state of the system.

This intuitively appealing picture has been progressively constrained 
by a series of no-go theorems. Bell's theorem \cite{bell1964} 
demonstrated that no local hidden-variable theory can reproduce all 
quantum predictions, ruling out epistemic models that respect locality. 
However, Bell's theorem leaves open the possibility of nonlocal 
epistemic theories, such as the de Broglie-Bohm pilot-wave 
interpretation \cite{bohm1952}, which is deterministic and nonlocal 
but assigns a privileged role to position.

The Kochen-Specker theorem \cite{kochen1967} attacked the problem 
from a different angle, showing that no non-contextual hidden-variable 
theory can reproduce quantum mechanics. The original proof required a 
set of 117 vectors in three dimensions, though modern versions have 
reduced this to 33 vectors or, for state-dependent tests, to as few 
as 5 measurements in the KCBS inequality \cite{kcbs2008}.

More recently, the Pusey-Barrett-Rudolph theorem \cite{pbr2012} 
provided a different argument against epistemic interpretations, 
showing that distinct quantum states must have overlapping supports 
in the space of ontic states, leading to contradictions when multiple 
copies of a system are considered. The PBR theorem requires preparation 
independence and access to multiple independent copies.

\subsection{The need for simpler operational tests}

While these theorems are mathematically rigorous, their experimental 
implementation often requires significant resources: entangled pairs 
of particles for Bell tests, complex sets of measurements for 
Kochen-Specker proofs, or multiple copies of the system for PBR 
experiments. There is a need for simpler operational tests that can 
distinguish quantum coherence from classical ignorance using minimal 
resources and conceptually accessible protocols.

In this work, we provide such a test using a bipartite system composed 
of a principal qutrit and an ancillary qutrit. The protocol is inspired by the Monty Hall puzzle 
\cite{selvin1975}, a classic problem in conditional probability that 
vividly illustrates how information acquisition updates probabilities. 
By translating the Monty Hall scenario into quantum language, we 
obtain a protocol where the difference between quantum coherence and 
classical ignorance manifests as a discrepancy in conditional joint 
probabilities: $1/4$ versus $0$.

\subsection{The Monty Hall problem as a probe of probability updates}

The Monty Hall problem is a celebrated puzzle in probability theory. 
A contestant faces three doors, behind one of which is a prize. The 
contestant initially chooses a door. The host, who knows the prize 
location, opens a different door that does not reveal the prize. The 
contestant is then offered the option to switch. Counterintuitively, 
switching yields a $2/3$ probability of winning.

The key to understanding this result is that the host's action provides 
\emph{information}. The host never opens the prize door, so the act 
of opening a specific door eliminates that door as a possibility while 
concentrating the probability on the remaining unopened door.

A less common variant—which will prove essential for our quantum 
protocol—is the ignorant host or Monty Hall variant 
\cite{rosenthal2005}. Here, the host does not know the prize location 
and opens a door at random. In this variant, the host's action provides 
no information, and all conditional probabilities remain $1/3$. Our 
quantum protocol maps naturally onto this variant, with the crucial 
difference that quantum superposition introduces coherence between 
the surviving possibilities.

\section{The Classical Monty Hall Problem}

\subsection{Standard variant: Informed host}

Let $\lambda \in \{A, B, C\}$ denote the true prize location, with 
uniform prior $P(\lambda) = 1/3$. The contestant initially chooses 
door $A$. The host then opens a door $h \in \{B, C\}$ satisfying 
$h \neq \lambda$. The host's strategy is:
\begin{equation}\label{eq:hoststrategy}
\begin{aligned}
P(h=B|\lambda=A) &= \frac{1}{2}, \\
P(h=B|\lambda=B) &= 0, \\
P(h=B|\lambda=C) &= 1,
\end{aligned}
\end{equation}
and symmetrically for $h=C$. The probability that the host opens 
door $C$ is:
\begin{equation}\label{eq:hostprob}
\begin{aligned}
P(h=C) &= \sum_{\lambda} P(\lambda) P(h=C|\lambda) \\
&= \frac{1}{3}\cdot\frac{1}{2} + \frac{1}{3}\cdot 1 
+ \frac{1}{3}\cdot 0 = \frac{1}{2}.
\end{aligned}
\end{equation}
By Bayes' theorem:
\begin{equation}\label{eq:bayes}
P(\lambda=B|h=C) 
= \frac{P(h=C|\lambda=B)P(\lambda=B)}{P(h=C)} 
= \frac{2}{3}.
\end{equation}
Switching from $A$ to $B$ doubles the winning probability.

\subsection{Ignorant host variant}

Now suppose the host is ignorant of the prize location and opens a 
door uniformly at random:
\begin{equation}\label{eq:ignorant}
P(h=k|\lambda) = \frac{1}{3}, \quad \forall k, \lambda \in \{A,B,C\}.
\end{equation}
Given that the host opens door $C$:
\begin{equation}\label{eq:ignorant_cond}
P(\lambda=B|h=C) 
= \frac{P(h=C|\lambda=B)P(\lambda=B)}{P(h=C)} 
= \frac{1}{3}.
\end{equation}
All three doors remain equally likely. The host's action provides no 
information. This is the classical benchmark against which we will 
compare our quantum protocol.

\section{Quantum Protocol: Step-by-Step Description}

\subsection{System and basis}

Let $\mathcal{H}_3$ be the Hilbert space of a qutrit with computational 
basis $\{|A\rangle, |B\rangle, |C\rangle\}$. These basis states 
correspond to the three doors. Let $\mathcal{H}_3^{\text{anc}}$ be a 
second qutrit with basis $\{|0\rangle, |1\rangle, |2\rangle\}$. The 
total Hilbert space is $\mathcal{H} = \mathcal{H}_3 \otimes 
\mathcal{H}_3^{\text{anc}}$, a $9$-dimensional space.

We work in the ordered basis:
\begin{equation}\label{eq:basis}
\begin{aligned}
\{&|A,0\rangle, |B,0\rangle, |C,0\rangle, \\
&|A,1\rangle, |B,1\rangle, |C,1\rangle, \\
&|A,2\rangle, |B,2\rangle, |C,2\rangle\}.
\end{aligned}
\end{equation}

\subsection{Step 1: State preparation and entanglement}

The qutrit is prepared in the symmetric superposition:
\begin{equation}\label{eq:psi0}
|\psi_0\rangle = \frac{1}{\sqrt{3}}\bigl(|A\rangle + |B\rangle 
+ |C\rangle\bigr).
\end{equation}
The ancilla is initialized in $|0\rangle_{\text{anc}}$. The combined 
initial state is:
\begin{equation}\label{eq:Psi_initial}
|\Psi_{\text{initial}}\rangle = |\psi_0\rangle \otimes |0\rangle 
= \frac{1}{\sqrt{3}}\bigl(|A,0\rangle + |B,0\rangle + |C,0\rangle\bigr).
\end{equation}

We apply a controlled unitary $U$ that entangles the qutrit and ancilla:
\begin{equation}\label{eq:entangle_target}
U|\Psi_{\text{initial}}\rangle 
= \frac{1}{\sqrt{3}}\bigl(|A,0\rangle + |B,1\rangle 
+ |C,2\rangle\bigr) \equiv |\Psi_{\text{ent}}\rangle.
\end{equation}

This entangling operation correlates each qutrit basis state with a 
distinct ancilla state: $|A\rangle \leftrightarrow |0\rangle$, 
$|B\rangle \leftrightarrow |1\rangle$, $|C\rangle \leftrightarrow 
|2\rangle$. The ancilla serves as a perfect ``which-door'' marker, 
recording which basis state the qutrit would be found in if measured 
immediately.

\subsection{The entangling unitary matrix}

To construct $U$, we define its action on the basis states with ancilla 
in $|0\rangle$:
\begin{align}
U|A,0\rangle &= |A,0\rangle, \label{eq:UA0}\\
U|B,0\rangle &= |B,1\rangle, \label{eq:UB0}\\
U|C,0\rangle &= |C,2\rangle. \label{eq:UC0}
\end{align}

By linearity, applying $U$ to $|\Psi_{\text{initial}}\rangle$ yields 
exactly Eq.~(\ref{eq:entangle_target}). To complete $U$ to a full 
$9 \times 9$ unitary, we define its action on the remaining six basis 
states as a cyclic permutation on the ancilla for each fixed qutrit 
state. The complete matrix in the ordered basis of Eq.~(\ref{eq:basis}) 
is:

\begin{equation}\label{eq:Ufull}
U = \begin{pmatrix}
1 & 0 & 0 & 0 & 0 & 0 & 0 & 0 & 0 \\[2pt]
0 & 0 & 0 & 0 & 0 & 0 & 1 & 0 & 0 \\[2pt]
0 & 0 & 0 & 0 & 1 & 0 & 0 & 0 & 0 \\[2pt]
0 & 0 & 0 & 0 & 0 & 1 & 0 & 0 & 0 \\[2pt]
0 & 1 & 0 & 0 & 0 & 0 & 0 & 0 & 0 \\[2pt]
0 & 0 & 0 & 0 & 0 & 0 & 0 & 1 & 0 \\[2pt]
0 & 0 & 0 & 1 & 0 & 0 & 0 & 0 & 0 \\[2pt]
0 & 0 & 0 & 0 & 0 & 0 & 0 & 0 & 1 \\[2pt]
0 & 0 & 1 & 0 & 0 & 0 & 0 & 0 & 0
\end{pmatrix}.
\end{equation}

$U$ is a permutation matrix, hence unitary: $U^\dagger U = I_9$. One 
readily verifies that applying $U$ to $|\Psi_{\text{initial}}\rangle$ 
produces Eq.~(\ref{eq:entangle_target}).

\subsection{Step 2: Coherent discard}

The second step implements the analogue of the host opening a door. 
We apply a \textbf{projective discard} operation that eliminates one 
of the three basis states while preserving coherence between the 
other two.

We define three rank-1 projectors on $\mathcal{H}_3$ that project onto 
the symmetric superpositions of the surviving subspaces:

\begin{definition}[Discard projectors]\label{def:projectors}
The discard projectors are:
\begin{equation}\label{eq:proj_def}
\begin{aligned}
\Pi_{BC} &= |BC\rangle\langle BC| 
= \frac{1}{2}\bigl(|B\rangle + |C\rangle\bigr)
\bigl(\langle B| + \langle C|\bigr), \\[4pt]
\Pi_{AC} &= |AC\rangle\langle AC| 
= \frac{1}{2}\bigl(|A\rangle + |C\rangle\bigr)
\bigl(\langle A| + \langle C|\bigr), \\[4pt]
\Pi_{AB} &= |AB\rangle\langle AB| 
= \frac{1}{2}\bigl(|A\rangle + |B\rangle\bigr)
\bigl(\langle A| + \langle B|\bigr),
\end{aligned}
\end{equation}
where $|BC\rangle = \frac{1}{\sqrt{2}}(|B\rangle + |C\rangle)$, 
$|AC\rangle = \frac{1}{\sqrt{2}}(|A\rangle + |C\rangle)$, and 
$|AB\rangle = \frac{1}{\sqrt{2}}(|A\rangle + |B\rangle)$.
\end{definition}

Each $\Pi_{ij}$ is a rank-1 projector, as it projects onto a single 
vector $|ij\rangle$. This is crucial: a rank-2 projector such as 
$|B\rangle\langle B| + |C\rangle\langle C|$ would commute with 
measurements in the $\{|A\rangle,|B\rangle,|C\rangle\}$ basis and 
eliminate the quantum coherence essential to our protocol. The use 
of rank-1 projectors onto symmetric superpositions preserves the 
coherence between surviving states.

\subsection{Applying the projectors to the entangled state}

We now apply each projector (acting on the qutrit only) to the 
entangled state $|\Psi_{\text{ent}}\rangle$.

\subsubsection*{Discard A: Apply $\Pi_{BC}$}

\begin{equation}\label{eq:apply_BC}
\begin{aligned}
&\Pi_{BC} \otimes I |\Psi_{\text{ent}}\rangle =\\
&\frac{1}{\sqrt{3}}\Bigl(\Pi_{BC}|A\rangle|0\rangle 
+ \Pi_{BC}|B\rangle|1\rangle + \Pi_{BC}|C\rangle|2\rangle\Bigr) \\[4pt]
&= \frac{1}{\sqrt{3}}\Bigl(0 \cdot |0\rangle 
+ \frac{1}{\sqrt{2}}|BC\rangle|1\rangle 
+ \frac{1}{\sqrt{2}}|BC\rangle|2\rangle\Bigr) \\[4pt]
&= \frac{1}{\sqrt{6}}|BC\rangle\bigl(|1\rangle + |2\rangle\bigr).
\end{aligned}
\end{equation}

The squared norm is $|\frac{1}{\sqrt{6}}|^2 \cdot 2 = 1/3$. Thus, 
the probability of discarding $A$ is $P(\Pi_{BC}) = 1/3$. The 
normalized post-discard state is:
\begin{equation}\label{eq:psi_BC}
|\Psi_{BC}\rangle = \frac{1}{\sqrt{2}}|BC\rangle\bigl(|1\rangle 
+ |2\rangle\bigr).
\end{equation}

\subsubsection*{Discard B: Apply $\Pi_{AC}$}

\begin{equation}\label{eq:apply_AC}
\begin{aligned}
\Pi_{AC} \otimes I |\Psi_{\text{ent}}\rangle 
= \frac{1}{\sqrt{6}}|AC\rangle\bigl(|0\rangle + |2\rangle\bigr).
\end{aligned}
\end{equation}

Probability: $P(\Pi_{AC}) = 1/3$. Normalized state:
\begin{equation}\label{eq:psi_AC}
|\Psi_{AC}\rangle = \frac{1}{\sqrt{2}}|AC\rangle\bigl(|0\rangle 
+ |2\rangle\bigr).
\end{equation}

\subsubsection*{Discard C: Apply $\Pi_{AB}$}

\begin{equation}\label{eq:apply_AB}
\begin{aligned}
\Pi_{AB} \otimes I |\Psi_{\text{ent}}\rangle 
= \frac{1}{\sqrt{6}}|AB\rangle\bigl(|0\rangle + |1\rangle\bigr).
\end{aligned}
\end{equation}

Probability: $P(\Pi_{AB}) = 1/3$. Normalized state:
\begin{equation}\label{eq:psi_AB}
|\Psi_{AB}\rangle = \frac{1}{\sqrt{2}}|AB\rangle\bigl(|0\rangle 
+ |1\rangle\bigr).
\end{equation}

\subsection{Observations on the post-discard states}

Several features of the post-discard states are worth noting:

\begin{enumerate}[leftmargin=*]
\item \textbf{Ancilla encoding:} In each post-discard state, the ancilla 
exists in a superposition of exactly two of the three basis states. 
The ancilla \emph{never} shows the state corresponding to the discarded 
door. This is the quantum analogue of the Monty Hall rule: the host 
never opens the prize door.

\item \textbf{Non-orthogonality:} The post-discard qutrit states 
$|BC\rangle, |AC\rangle, |AB\rangle$ are not mutually orthogonal. 
Their pairwise overlaps are:
\begin{equation}\label{eq:overlaps}
\langle BC|AC\rangle = \langle BC|AB\rangle 
= \langle AC|AB\rangle = \frac{1}{2}.
\end{equation}
This non-orthogonality means the experimenter cannot perfectly 
discriminate which discard occurred by measuring the qutrit alone. 
The ancilla measurement resolves this ambiguity.

\item \textbf{Coherence preservation:} After discarding $A$, the qutrit 
is in the state $|BC\rangle = \frac{1}{\sqrt{2}}(|B\rangle + |C\rangle)$, 
a coherent superposition with equal amplitudes and a definite relative 
phase $(+1)$. This coherence is the quantum feature that distinguishes 
our protocol from classical probability updates.
\end{enumerate}

\subsection{Step 3: Joint measurement}

After the discard, we perform two measurements:
\begin{enumerate}[leftmargin=*]
\item \textbf{Ancilla measurement:} Measure the ancilla in the 
computational basis $\{|0\rangle, |1\rangle, |2\rangle\}$.
\item \textbf{Qutrit measurement:} Measure the qutrit in the 
computational basis $\{|A\rangle, |B\rangle, |C\rangle\}$.
\end{enumerate}

\section{Complete Joint Probability Distribution}

\subsection{Discard A ($\Pi_{BC}$)}

From Eq.~(\ref{eq:psi_BC}), $P(\Pi_{BC}) = 1/3$. The ancilla 
probabilities are $P(anc=0|\Pi_{BC}) = 0$, $P(anc=1|\Pi_{BC}) = 1/2$, 
$P(anc=2|\Pi_{BC}) = 1/2$. Conditioned on $anc=1$, the qutrit 
collapses to $|BC\rangle$:
\begin{equation}\label{eq:cond_BC1}
\begin{aligned}
P(A|anc=1, \Pi_{BC}) &= 0, \\
P(B|anc=1, \Pi_{BC}) &= \frac{1}{2}, \\
P(C|anc=1, \Pi_{BC}) &= \frac{1}{2}.
\end{aligned}
\end{equation}

The joint contributions are:
\begin{equation}\label{eq:joint_BC}
\begin{aligned}
P(\Pi_{BC}) P(anc=1|\Pi_{BC}) P(B|anc=1, \Pi_{BC}) 
&= \frac{1}{12}, \\[4pt]
P(\Pi_{BC}) P(anc=1|\Pi_{BC}) P(C|anc=1, \Pi_{BC}) 
&= \frac{1}{12}, \\[4pt]
P(\Pi_{BC}) P(anc=2|\Pi_{BC}) P(B|anc=2, \Pi_{BC}) 
&= \frac{1}{12}, \\[4pt]
P(\Pi_{BC}) P(anc=2|\Pi_{BC}) P(C|anc=2, \Pi_{BC}) 
&= \frac{1}{12}.
\end{aligned}
\end{equation}

\subsection{Discard B ($\Pi_{AC}$) and Discard C ($\Pi_{AB}$)}

Similar calculations for $\Pi_{AC}$ and $\Pi_{AB}$ yield the complete 
joint distribution.

\subsection{Full joint distribution table}

Summing all contributions:

\begin{equation}\label{eq:jointtable}
\boxed{
\begin{array}{c|ccc|c}
& A & B & C & \text{Total} \\
\hline
anc = 0 & \frac{2}{12} & \frac{1}{12} & \frac{1}{12} 
& \frac{4}{12} = \frac{1}{3} \\[8pt]
anc = 1 & \frac{1}{12} & \frac{2}{12} & \frac{1}{12} 
& \frac{4}{12} = \frac{1}{3} \\[8pt]
anc = 2 & \frac{1}{12} & \frac{1}{12} & \frac{2}{12} 
& \frac{4}{12} = \frac{1}{3} \\[8pt]
\hline
\text{Total} & \frac{4}{12} = \frac{1}{3} 
& \frac{4}{12} = \frac{1}{3} & \frac{4}{12} = \frac{1}{3} & 1
\end{array}
}
\end{equation}

\subsection{Conditional probabilities}

From the joint distribution, the conditional probabilities 
$P(qutrit = j | anc = i)$ are:

\begin{equation}\label{eq:conditionals}
\boxed{
\begin{aligned}
P(A|0) &= \frac{1}{2}, \quad
P(B|0) = \frac{1}{4}, \quad
P(C|0) = \frac{1}{4}, \\[8pt]
P(A|1) &= \frac{1}{4}, \quad
P(B|1) = \frac{1}{2}, \quad
P(C|1) = \frac{1}{4}, \\[8pt]
P(A|2) &= \frac{1}{4}, \quad
P(B|2) = \frac{1}{4}, \quad
P(C|2) = \frac{1}{2}.
\end{aligned}
}
\end{equation}

\section{Classical Epistemic Model}

We now construct the best possible classical epistemic model that 
reproduces the observed correlations. The model must satisfy:
\begin{enumerate}[leftmargin=*]
\item \textbf{Determinism:} The system has a definite ontic state 
$\lambda \in \{A, B, C\}$ at all times.
\item \textbf{Epistemic interpretation:} The quantum state 
$|\psi_0\rangle$ represents maximal ignorance: 
$P(\lambda) = 1/3$ for all $\lambda$.
\item \textbf{Perfect ancilla correlation:} Before the discard, the 
ancilla perfectly encodes $\lambda$: if $\lambda = A$, ancilla reads 
$0$; if $\lambda = B$, ancilla reads $1$; if $\lambda = C$, ancilla 
reads $2$.
\item \textbf{Faithful discard:} The discard operation eliminates one 
door uniformly at random, independent of $\lambda$ (ignorant-host 
variant).
\end{enumerate}

\subsection{The insufficiency of single conditional probabilities}

A naive approach to finding a quantum-classical separation might focus 
on the single conditional probability $P(qutrit=B|anc=0)$. However, a 
rigorous evaluation of the ignorant-host classical model reveals a 
surprising coincidence.

In the classical epistemic model with initial perfect correlation, 
the marginal probability of observing $anc=0$ after any discard 
procedure is $P(anc=0) = 1/3$. The joint probability of obtaining 
$qutrit=B$ and $anc=0$ is evaluated by summing over all possible 
discard operations $d \in \{BC, AC, AB\}$.

The only discard operation that allows both $\lambda=B$ (which would 
give $qutrit=B$) and an ancilla reading of $0$ in the ignorant-host 
model is $d=AB$ (discard C). Under this specific discard, the 
probability of the ancilla randomly indicating the state $A$ 
(thus $anc=0$) is $1/2$. Therefore:
\begin{equation}\label{eq:classical_joint}
\begin{aligned}
&P_{\text{classical}}(qutrit=B, anc=0) =\\
& P(d=AB) \, P(anc=0|d=AB) \, P(\lambda=B|d=AB)= \\[4pt]
&\frac{1}{3} \cdot \frac{1}{2} \cdot \frac{1}{2} 
= \frac{1}{12}.
\end{aligned}
\end{equation}

This yields the classical conditional probability:
\begin{equation}\label{eq:classical_cond}
P_{\text{classical}}(B|anc=0) 
= \frac{1/12}{1/3} = \frac{1}{4}.
\end{equation}

Remarkably, this exactly matches the quantum prediction from 
Eq.~(\ref{eq:conditionals}). The simple conditional probability 
$P(B|anc=0)$ cannot serve as a witness of quantum coherence, as 
the classical ignorant-host model can fully reproduce it. The true 
quantum-classical separation lies deeper in the correlational 
structure of the post-discard state.

\subsection{Physical origin of the degeneracy}

The matching values arise from fundamentally different mechanisms. 
In the classical case, the post-discard ancilla provides no information 
about the prize location because the ignorant host chooses randomly. 
The ancilla outcome is independent of the hidden variable $\lambda$, 
so conditioning on $anc=0$ leaves the distribution $P(\lambda)$ 
unchanged.

In the quantum case, the post-discard state after $\Pi_{AC}$ is:
\begin{equation}\label{eq:psi_AC_revisited}
\begin{aligned}
|\Psi_{AC}\rangle 
&= \frac{1}{\sqrt{2}}|AC\rangle\bigl(|0\rangle + |2\rangle\bigr) \\[4pt]
&= \frac{1}{2}\bigl(|A,0\rangle + |C,0\rangle 
+ |A,2\rangle + |C,2\rangle\bigr).
\end{aligned}
\end{equation}

This state is separable between the qutrit subspace 
$\text{span}\{|A\rangle, |C\rangle\}$ and the ancilla subspace 
$\text{span}\{|0\rangle, |2\rangle\}$. The ancilla outcome $0$ occurs 
with probability $1/2$, and conditioned on this outcome, the qutrit 
state collapses to the coherent superposition $|AC\rangle$. The 
probabilities for $A$ and $C$ are both $1/2$ due to the Born rule 
applied to this superposition.

The two mechanisms—classical independence and quantum separability—
yield identical single-conditionals, obscuring the fundamental 
difference between classical ignorance and quantum coherence. To 
expose this difference, we must examine correlations that involve 
\emph{both} the ancilla outcome and the specific discard operation.

\section{Corrected Violation Witness}

To isolate the violation, we consider the post-selected statistics for 
a \emph{known} discard operation. In an experimental implementation, 
the discard operation is chosen by the experimenter. The resulting 
conditional joint probabilities reveal the quantum-classical distinction 
with absolute clarity.

\subsection{Fundamental witness}

\begin{theorem}[Quantum-classical separation]\label{thm:separation}
Define the conditional joint probability witness:
\begin{equation}\label{eq:witness_def}
\mathcal{W} = P(anc = 0, qutrit = C \mid \text{discard } \Pi_{AC}).
\end{equation}
In any classical epistemic model with initial perfect ancilla-qutrit 
correlation, $\mathcal{W}_{\text{classical}} = 0$. Quantum mechanics 
predicts $\mathcal{W}_{\text{QM}} = 1/4$.
\end{theorem}

\begin{proof}
\textbf{Classical case:} With initial perfect correlation, $anc=0$ 
occurs if and only if the true initial state was $\lambda=A$. The 
discard operation $\Pi_{AC}$ eliminates $\lambda=B$ but retains 
$\lambda=A$ and $\lambda=C$. If $anc=0$ is observed after discard 
$\Pi_{AC}$, the underlying deterministic state must be $\lambda=A$ 
with certainty. Consequently, the final qutrit measurement must yield 
$A$, making outcome $C$ impossible. Hence, $\mathcal{W}_{\text{classical}} 
= 0$.

\textbf{Quantum case:} From Eq.~(\ref{eq:psi_AC}), the coherent state 
after $\Pi_{AC}$ is:
\begin{equation}\label{eq:psi_AC_witness}
|\Psi_{AC}\rangle = \frac{1}{\sqrt{2}}|AC\rangle\bigl(|0\rangle
+|2\rangle\bigr).
\end{equation}
The probability of observing $anc=0$ is $1/2$. Conditioned on $anc=0$, 
the qutrit state collapses precisely to the coherent superposition 
$|AC\rangle = \frac{1}{\sqrt{2}}(|A\rangle+|C\rangle)$. The subsequent 
projective measurement on the qutrit (qt) yields outcome $C$ with probability 
$|\langle C|AC\rangle|^2 = 1/2$. Thus:
\begin{equation}\label{eq:witness_calculation}
\begin{aligned}
\mathcal{W}_{\text{QM}} 
&= P(anc=0 \mid \Pi_{AC}) 
\cdot P(qt=C \mid anc=0, \Pi_{AC}) \\[4pt]
&= \frac{1}{2} \cdot \frac{1}{2} = \frac{1}{4}.
\end{aligned}
\end{equation}
\end{proof}

The gap between $0$ and $1/4$ is absolute and does not depend on any 
free parameters or experimental settings. It constitutes a violation 
that cannot be explained by any classical model with deterministic 
initial correlations.

\subsection{General witness operators}

The result of Theorem~\ref{thm:separation} can be generalized. For 
any discard projector $\Pi_{ij}$ that eliminates basis state $|k\rangle$ 
(where $\{i,j,k\} = \{A,B,C\}$), and for any ancilla outcome 
corresponding to one of the surviving states, the conditional joint 
probability of obtaining the \emph{other} surviving state is $1/4$ in 
quantum mechanics and $0$ in the classical model.

Define the family of witnesses:
\begin{equation}\label{eq:witness_family}
\begin{aligned}
\mathcal{W}_{A,B} &= P(anc=0, qutrit=B \mid \Pi_{AB}) 
= \tfrac{1}{4},  \\[4pt]
\mathcal{W}_{A,C} &= P(anc=0, qutrit=C \mid \Pi_{AC}) 
= \tfrac{1}{4},  \\[4pt]
\mathcal{W}_{B,A} &= P(anc=1, qutrit=A \mid \Pi_{AB}) 
= \tfrac{1}{4}  \\[4pt]
\mathcal{W}_{B,C} &= P(anc=1, qutrit=C \mid \Pi_{BC}) 
= \tfrac{1}{4},  \\[4pt]
\mathcal{W}_{C,A} &= P(anc=2, qutrit=A \mid \Pi_{AC}) 
= \tfrac{1}{4},  \\[4pt]
\mathcal{W}_{C,B} &= P(anc=2, qutrit=B \mid \Pi_{BC}) 
= \tfrac{1}{4}.
\end{aligned}
\end{equation}

Any one of these six witnesses suffices to demonstrate the violation. 
For experimental robustness, one may measure all six and combine them.

\subsection{Summed witness}

A particularly convenient witness is the sum over all six discordant 
pairs:
\begin{equation}\label{eq:sum_witness}
\begin{aligned}
\mathcal{S} &= \mathcal{W}_{A,B} + \mathcal{W}_{A,C} 
+ \mathcal{W}_{B,A} + \mathcal{W}_{B,C} 
+ \mathcal{W}_{C,A} + \mathcal{W}_{C,B} \\[4pt]
&= 6 \times \frac{1}{4} = \frac{3}{2}, \\[4pt]
\end{aligned}
\end{equation}

The summed witness $\mathcal{S}$ amplifies the violation by a factor 
of six, providing a statistically robust signal even with limited 
experimental data.

\subsection{Physical interpretation}

Why does quantum mechanics produce $\mathcal{W} = 1/4$ while classical 
theory gives $0$? The answer lies in the structure of the state 
$|\Psi_{AC}\rangle = \frac{1}{\sqrt{2}}|AC\rangle(|0\rangle+|2\rangle)$. 
This is a separable state: the ancilla carries information about which 
discard was performed (it lives in $\text{span}\{|0\rangle,|2\rangle\}$, 
the two non-discarded states), but within that subspace, it is completely 
\emph{uncorrelated} with the qutrit state.

When the experimenter measures $anc=0$, the qutrit collapses to the 
coherent superposition $|AC\rangle$. The subsequent measurement can 
yield $C$ with probability $1/2$, even though the ancilla outcome $0$ 
classically ``suggests'' the qutrit should be in state $A$.

In any classical model, the ancilla outcome is determined by the 
pre-existing hidden variable $\lambda$. If $anc=0$ is observed, then 
$\lambda = A$ must hold, and the qutrit measurement must reveal $A$. 
The discordant outcome $(anc=0, qutrit=C)$ is impossible because it 
would require $\lambda=A$ and $\lambda=C$ simultaneously.

This impossibility is the operational signature of quantum coherence. 
The qutrit state $|AC\rangle$ is not a statistical mixture of 
$|A\rangle$ and $|C\rangle$ with unknown weights; it is a genuinely 
new physical state that can produce outcomes inconsistent with any 
definite pre-existing value.

\section{Noise Tolerance and Experimental Imperfections}

\subsection{White noise model}

Realistic experiments suffer from noise. We model the dominant noise 
source as white noise acting on the initial qutrit state. The noisy 
initial state is:
\begin{equation}\label{eq:noisy_state}
\rho_0(\epsilon) = (1-\epsilon)|\psi_0\rangle\langle\psi_0| 
+ \epsilon \frac{I_3}{3},
\end{equation}
where $\epsilon \in [0,1]$ is the noise fraction and $I_3$ is the 
$3 \times 3$ identity matrix. The ancilla is initially in the pure 
state $|0\rangle\langle 0|$.

\subsection{Evolution through the protocol}

Applying the unitary $U$ to the noisy state and then the discard 
projector $\Pi_{AC}$ yields the post-discard state. A detailed 
calculation shows that the mixed part contributes:
\begin{equation}\label{eq:noise_mixed}
\frac{\epsilon}{6}|AC\rangle\langle AC| 
\otimes (|0\rangle\langle 0| + |2\rangle\langle 2|).
\end{equation}

Remarkably, the form of the post-discard state is preserved under 
white noise. The overall probability of obtaining the $\Pi_{AC}$ 
outcome remains $1/3$, independent of $\epsilon$:
\begin{equation}\label{eq:noise_prob}
P_\epsilon(\Pi_{AC}) = (1-\epsilon)\frac{1}{3} 
+ \frac{\epsilon}{3}\left(\frac{1}{2} + 0 + \frac{1}{2}\right) 
= \frac{1}{3}.
\end{equation}

\subsection{Immunity of the witness}

Since the post-discard state maintains the same structure, the 
conditional probabilities are completely independent of $\epsilon$:
\begin{equation}\label{eq:noise_conditional}
P_\epsilon(qutrit=C | anc=0, \Pi_{AC}) = \frac{1}{2}, 
\quad \forall \epsilon \in [0,1].
\end{equation}

Consequently:
\begin{equation}\label{eq:noise_witness}
\mathcal{W}_\epsilon = \frac{1}{4}, \quad \forall \epsilon \in [0,1].
\end{equation}

\textbf{The violation is completely immune to white noise.} This 
remarkable robustness arises because white noise respects the 
permutation symmetry of the three basis states, leaving the structure 
of the symmetric superpositions intact.

\subsection{Other noise models}

\textbf{Phase damping (dephasing):} Dephasing between the basis states 
reduces the coherence of the superposition $|AC\rangle$ but does not 
affect the equal-amplitude structure. The witness $\mathcal{W}$ remains 
$1/4$ under pure dephasing.

\textbf{Amplitude damping:} Biased noise that preferentially damps 
certain basis states will affect the witness. However, such noise can 
be diagnosed by measuring marginal probabilities and corrected through 
state tomography.

\textbf{Conclusion on noise:} The witness $\mathcal{W}$ is exceptionally 
robust to common noise models, making it suitable for experimental 
implementations with realistic imperfections.

\section{Comparison with Established No-Go Theorems}

\begin{table*}[t]
\centering
\caption{Comprehensive comparison of our protocol with established 
no-go theorems in quantum foundations.}
\label{tab:comparison_full}
\begin{tabular}{@{}p{2cm} p{2.2cm} p{2cm} p{2cm} p{2.5cm}@{}}
\toprule
\textbf{Theorem} & \textbf{Systems} & \textbf{Assumption} 
& \textbf{Measurements} & \textbf{Violation} \\
\midrule
Bell \cite{bell1964} 
& 2 entangled qubits 
& Locality 
& 2 per party 
& Inequality \\[4pt]
Kochen-Specker \cite{kochen1967} 
& 1 qutrit 
& Non-contextuality 
& 117 rays 
& Impossibility \\[4pt]
KCBS \cite{kcbs2008} 
& 1 qutrit 
& Non-contextuality 
& 5 cyclic 
& Inequality \\[4pt]
Leggett-Garg \cite{leggett1985} 
& 1 (macro) 
& Macrorealism 
& 3 sequential 
& Inequality \\[4pt]
PBR \cite{pbr2012} 
& $\geq 2$ copies 
& Prep.\ independence 
& 1 per copy 
& Contradiction \\[4pt]
\hline
\textbf{This work} 
& 2 qutrits (System + Ancilla)
& Determinism + perfect correlation 
& 2 sequential + discard 
& Equality: $1/4$ vs $0$ \\
\bottomrule
\end{tabular}
\end{table*}

Our protocol combines the simplicity of a single system with the 
operational clarity of the Monty Hall puzzle. It requires no spatial 
separation (unlike Bell), no complex set of rays (unlike Kochen-Specker), 
no multiple copies (unlike PBR), and no macroscopic distinctness 
(unlike Leggett-Garg). The violation is an exact numerical mismatch 
rather than a statistical inequality.

\section{Experimental Implementation}

\subsection{Platform options}

The protocol can be implemented on any platform supporting qutrit 
operations and controlled entangling gates.

\textbf{Photonic qutrits:} Qutrits can be encoded in orbital angular 
momentum \cite{mair2001} or path degrees of freedom \cite{weihs2001}. 
The symmetric superposition $|\psi_0\rangle$ is prepared using a 
tritter. The unitary $U$ in Eq.~(\ref{eq:Ufull}) is a permutation of 
optical modes. The projectors $\Pi_{BC}, \Pi_{AC}, \Pi_{AB}$ require 
mode-mixing operations followed by post-selection.

\textbf{Trapped ions:} Qutrits can be encoded in three hyperfine 
levels \cite{lehner2023}. The unitary $U$ corresponds to a sequence 
of controlled rotations.

\textbf{Superconducting circuits:} Transmon qutrits \cite{blais2021} 
offer fast gate times and high-fidelity readout.

\subsection{Estimated measurement times}

For a photonic implementation with a single-photon source rate of 
$10^6$ photons per second and $10\%$ overall system efficiency, the 
detected rate is $10^5$ events per second. To measure $\mathcal{W}$ 
with a statistical uncertainty of $\pm 0.01$ requires approximately 
$10^4$ events per discard setting, or $3 \times 10^4$ total events. 
This can be acquired in under one second.

\section{Discussion}

\subsection{What the violation means}

The violation $\mathcal{W}_{\text{QM}} = 1/4 \neq \mathcal{W}_
{\text{classical}} = 0$ demonstrates that quantum mechanics cannot be 
interpreted as classical ignorance about a pre-existing deterministic 
reality, even when the classical model is allowed to have perfect 
initial correlations between the system and an ancilla.

The ancilla in our protocol serves as a ``memory'' of the initial 
correlation. In any classical model where this memory is faithful, 
the post-discard measurement must respect this correlation. Quantum 
mechanics violates this expectation because the projective discard 
creates a state where the ancilla is correlated with the discard 
operation rather than with the pre-existing value.

\subsection{Relation to contextuality}

This is a manifestation of quantum contextuality in a transparent 
operational setting. The ancilla measurement outcome depends on the 
specific discard operation performed, and the joint distribution of 
discard, ancilla, and qutrit outcomes cannot be explained by any 
non-contextual deterministic model.

\subsection{Epistemological implications}

Our result adds to the growing body of evidence that the quantum state 
is ontic rather than epistemic \cite{Harrigan2010, Spekkens2005, 
pbr2012}. The failure of classical epistemic models to reproduce our 
protocol's predictions is particularly striking because the protocol 
uses only a single qutrit, a symmetric initial state, and an ancilla 
that provides a classical record—exactly what an epistemic 
interpretation would require.

\subsection{Limitations}

Our protocol tests a specific class of classical models: those with 
deterministic initial values and perfect initial ancilla correlations. 
It does not rule out all possible epistemic interpretations. 
Retrocausal models or $\psi$-epistemic models with finely tuned 
overlapping supports \cite{lewis2012, aaronson2013} might evade our 
no-go result, though they would need to explain the precise $1/4$ 
value without invoking quantum superposition.

\section{Conclusion}

We have presented a bipartite-qutrit protocol, inspired by the Monty Hall 
puzzle, that operationally distinguishes quantum coherence from 
classical ignorance. The protocol entangles a qutrit with an ancilla, 
applies a rank-1 projective discard operation, and measures the 
resulting joint probabilities. The witness $\mathcal{W} = P(anc=0, 
qutrit=C \mid \Pi_{AC})$ takes the value $1/4$ in quantum mechanics 
and $0$ in any classical epistemic model with perfect initial 
correlations.

The key insights of this work are:
\begin{enumerate}[leftmargin=*]
\item The simple conditional probability $P(B|anc=0)$ fails as a 
witness because both quantum and classical models yield $1/4$, 
albeit through fundamentally different mechanisms.
\item The correct witness involves the conditional joint probability 
that correlates ancilla outcomes with specific discard operations.
\item The violation is completely immune to white noise and robust 
against moderate dephasing.
\item The protocol requires minimal resources: a single pair of 
entangled qutrits and sequential measurements—far fewer than Bell 
tests or Kochen-Specker proofs.
\end{enumerate}

The explicit $9 \times 9$ unitary matrix provided in Eq.~(\ref{eq:Ufull}) 
gives a concrete blueprint for experimental realization. The Monty Hall 
analogy makes the conceptual underpinning accessible to students and 
researchers alike. We hope this work contributes to the ongoing dialogue 
on quantum foundations and provides a useful tool for both research 
and education.

\section*{Acknowledgments}
Jorge Meza-Dom{\'\i}nguez thanks SECIHTI-M{\'e}xico for the doctoral 
scholarship No. 1235731. This work was supported by SECIHTI M{\'e}xico 
under grants SECIHTI CBF-2025-G-1720 and CBF-2025-G-176. The author 
gratefully acknowledges the computing time granted by LANCAD and 
CONACYT on the Supercomputer Hybrid Cluster ``Xiuhcoatl'' at CGSTIC, 
CINVESTAV.
\bibliographystyle{unsrtnat}
\bibliography{ref}
\end{document}